\documentclass{birkmult}

\usepackage{latexsym,amsmath,amscd,amssymb,graphics,xypic}
\usepackage{enumerate}

\usepackage{graphicx,lscape}

\usepackage[colorlinks]{hyperref}
\usepackage{url}

\begin{document}

\newtheorem{theorem}{Theorem}[section]
\newtheorem{definition}[theorem]{Definition}
\newtheorem{lemma}[theorem]{Lemma}
\newtheorem{remark}[theorem]{Remark}
\newtheorem{proposition}[theorem]{Proposition}
\newtheorem{corollary}[theorem]{Corollary}
\newtheorem{example}[theorem]{Example}

\numberwithin{equation}{section}
\newcommand{\ep}{\varepsilon}
\newcommand{\R}{{\mathbb  R}}
\newcommand\C{{\mathbb  C}}
\newcommand\Q{{\mathbb Q}}
\newcommand\Z{{\mathbb Z}}
\newcommand{\N}{{\mathbb N}}

\title[The free rigid body dynamics: generalized versus classic]{The free rigid body dynamics: generalized versus classic}

\author{R\u{a}zvan M. Tudoran}

\address{The West University of Timi\c soara,\\
Faculty of Mathematics and Computer Science, \\
Department of Mathematics,\\
Blvd. Vasile P\^ arvan, No. 4,\\
300223 - Timi\c soara,\\
Romania.}

\email{tudoran@math.uvt.ro}



\date{February 22, 2011}
\dedicatory{}

\begin{abstract}
In this paper we analyze the normal forms of a general quadratic Hamiltonian system defined on the dual of the Lie algebra $\mathfrak{o}(K)$ of real $K$ - skew - symmetric matrices, where $K$ is an arbitrary $3\times 3$ real symmetric matrix. A consequence of the main results is that any first-order autonomous three-dimensional differential equation possessing two independent quadratic constants of motion which admits a positive/negative definite linear combination, is affinely equivalent to the classical "relaxed" free rigid body dynamics with linear controls.
\end{abstract}

\maketitle

\section{Introduction}
\label{section:one}

When thinking about quadratic and homogeneous Hamiltonian systems on the dual of a Lie algebra, the first example that comes to our mind is the system
describing the rotations of a free rigid body around its center of mass. This system was derived by Euler in 1758 (see e.g. \cite{euler}) then generalized by Poincar\'e (see \cite{poincare}), and then again later by Arnold (see e.g. \cite{arnoldcarte}) starting from the original $\mathfrak{so}(3)$ Lie algebra, to a general Lie algebra. The mathematical literature contains a huge 
amount of writings concerning Euler's equations, from their original form to the most general forms (see e.g. \cite{euler}, \cite{kowalewski}, \cite{abraham}, \cite{bogo1}, \cite{bogo2}, \cite{holm1}, \cite{holm2}, \cite{holm3}, \cite{ratiu},\cite{ec}, \cite{marsdenratiu}, \cite{rotorpendul}, \cite{tarama}). 

In this paper we study a large class of Hamiltonian systems defined on the dual of the Lie algebra $\mathfrak{o}(K)$ of real $K$ - skew-symmetric matrices, where $K$ is an arbitrary $3\times 3$ real symmetric matrix. More precisely, we consider the Hamiltonian systems on $(\mathfrak{o}(K))^*$, generated by quadratic Hamiltonian functions, and we show that under suitable conditions, the classical "relaxed" free rigid body dynamics, represents the normal form of a Hamiltonian system on $(\mathfrak{o}(K))^*$, generated by a quadratic and homogeneous Hamiltonian function, and respectively the classical "relaxed" free rigid body dynamics with three linear controls, represents the normal form of a Hamiltonian system on $(\mathfrak{o}(K))^*$, generated by a quadratic Hamiltonian function. As a consequence of these results we get that the extended free rigid body introduced in \cite{tarama} is equivalent to the classical "relaxed" free rigid body. Note that by classical "relaxed" free rigid body dynamics, we mean the dynamics generated by the classical Euler's equations of the free rigid body on $(\mathfrak{so}(3))^*$, where the Hamiltonian is supposed to be generated by an arbitrary $3\times 3$ diagonal real matrix. 

The most important consequence of the results of this paper is that any first-order autonomous three-dimensional differential equation possessing two independent quadratic constants of motion which admits a positive/negative definite linear combination, is affinely equivalent to the classical "relaxed" free rigid body dynamics with linear controls.

The structure of the paper is as follows: in the second section we prepare the framework of our study, by introducing the family of Hamiltonian systems to be analyzed, and recall some of the main geometrical properties of the Poisson configuration manifold. In the third section of this article one compute explicitly the normal form of a general quadratic and homogeneous Hamiltonian system on $(\mathfrak{o}(K))^{*}$. More precisely, one shows that if there exists $\alpha,\beta\in\mathbb{R}$ such that the symmetric real matrix $\alpha A+\beta K$ is positive definite, where $A$ denotes the symmetric matrix generating the quadratic and homogeneous Hamiltonian function $H_A$, then the Hamiltonian system $((\mathfrak{o}(K))^{*},\{\cdot,\cdot\}_{K},H_{A})$, is equivalent to the relaxed free rigid body dynamics. In the fourth section, one shows that the normal form of a general quadratic Hamiltonian system on $(\mathfrak{o}(K))^{*}$ in the case when there exists $\alpha,\beta\in\mathbb{R}$ such that the symmetric real matrix $\alpha A+\beta K$ is positive definite, is equivalent to the relaxed free rigid body dynamics with three linear controls. In the last section, one gives a unified and also generalized formulation of the results obtained in the previous sections, by analyzing the case of a general quadratic Hamiltonian system on a natural extension of $(\mathfrak{o}(K))^*$.

 For details on Poisson geometry and Hamiltonian dynamics, see, e.g. \cite{abraham}, \cite{arnoldcarte}, \cite{marsdenratiu}, \cite{holm1}, \cite{holm2}, \cite{holm3}, \cite{ratiurazvan}.

\section{Quadratic Hamiltonian systems on $(\mathfrak{o}(K))^*$}

As the purpose of this paper is to study the quadratic Hamiltonian systems on
$(\mathfrak{o}(K))^*$ from the Poisson geometry and dynamics point of view, the first step in this approach is to
prepare the geometric framework of the problem.

Let us recall first some generalities about the Lie algebra $\mathfrak{o}(K)$ of $3\times3$ real $K$-skew-symmetric matrices, and his dual space $(\mathfrak{o}(K))^*$, where $K$ is an arbitrary $3\times3$ real symmetric matrix. In the case when $K$ is nondegenerate, the Lie algebra $\mathfrak{o}(K):=\{A\in \mathfrak{gl}(3;\R):\ A^{T}K+KA=O_3\}$ is the Lie algebra of the Lie group $O(K):=\{A\in GL(3;\R):\ A^{T}KA=K \}$ of $K$-orthogonal $3\times3$ matrices. 

Let us now recall a Lie algebra isomorphism between $\mathfrak{o}(K)$ and $\R^3$. For details regarding this isomorphism, see e.g. \cite{holm1}, \cite{marsdenratiu}.
\begin{proposition}\label{pr.2.1}
The Lie algebras $\left(\mathfrak{o}(K),+,\cdot_{\R},[\cdot,\cdot]\right)$ and $\left(\R^3,+,\cdot_{\R},\times_{K}\right)$ are isomorphic, where $[\cdot,\cdot]$ is the commutator of matrices, and respectively $u\times_{K}v:=K(u\times v)$, for any $u,v\in\mathbb{R}^{3}$.
\end{proposition}
Note that for $K$ nonsingular, the above isomorphism between $S\in\mathfrak{o}(K)$ and $\bold{s}\in\mathbb{R}^3$, can be defined by using the equation $Su=\bold{s}\times Ku$, $u\in\mathbb{R}^3$.

An immediate consequence of the this proposition is that $(\mathfrak{o}(K))^*\cong(\R^3)^*\cong$ $\R^3$, viewed as a dual of a Lie algebra, it has a natural Poisson structure, namely the "minus" Lie-Poisson structure, which in this case proves to be generated by the Poisson bracket:
$$\{f,g\}_{K}:=-\nabla C_{K}\cdot(\nabla f\times\nabla g),$$ 
for any $f,g\in C^\infty(\R^3,\R)$, where the smooth function $C_{K}\in C^\infty(\R^3,\R)$ is given by
$$C_{K}(u):=\dfrac{1}{2}u^{T}Ku.$$
For more details regarding this bracket and the associated Lie-Poisson dynamics, see e.g. \cite{holm1}.  
\begin{remark}
It is not hard to see that the center of the Poisson algebra $C^\infty(\R^3,\R)$ is generated by the Casimir invariant
$C_{K}\in C^\infty(\R^3,\R)$, $C_{K}(u)=\dfrac{1}{2}u^TKu$.
\end{remark}
Hence, using the above isomorphism between $(\mathfrak{o}(K))^*$ and $\R^3$, a quadratic Hamiltonian system on
$(\mathfrak{o}(K))^*$ is isomorphically represented by a quadratic Hamiltonian system on $\R^3$ as follows:
$$(\R^3,\{\cdot,\cdot\}_{K},H_{(A,\bold{a})}),$$
where the Hamiltonian $H_{(A,\bold{a})}\in C^\infty(\R^3,\R)$ is given by $H_{(A,\bold{a})}(u):=\dfrac{1}{2}u^TAu+u^T\bold{a}$, with $A\in Sym(3)$ a given symmetric matrix, and $\bold{a}\in\mathbb{R}^{3}$.

\begin{remark}\label{rem.2.1}
A quadratic Hamiltonian system on $(\R^3,\{\cdot,\cdot\}_{K})$ is given by:
\begin{equation*}
\dot u=\nabla C_{K}(u)\times\nabla H_{(A,\bold{a})}(u),\ \ u\in\R^3,
\end{equation*}
or, equivalently:
\begin{equation}\label{eq.2.1}
\dot u=(Ku)\times(Au+\bold{a}),\ \ u\in\R^3.
\end{equation}
\end{remark}

\begin{remark}\label{rem.2.2}
For $\bold{a}=\bold{0}$, the Hamiltonian $H_{(A,\bold{0})}=:H_{A}$ becomes a quadratic and homogeneous Hamiltonian. The associated Hamiltonian system $(\R^3,\{\cdot,\cdot\}_{K},H_{A})$, is called a quadratic and homogeneous Hamiltonian system, and is given by:
\begin{equation*}
\dot u=\nabla C_{K}(u)\times\nabla H_{A}(u),\ \ u\in\R^3,
\end{equation*}
or, equivalently:
\begin{equation}\label{eq.2.2}
\dot u=(Ku)\times(Au),\ \ u\in\R^3.
\end{equation}
\end{remark}

\begin{remark}\label{rem.2.4}
For $K$ nondegenerate, the system (\ref{eq.2.1}) can also be regarded as a dynamical system resulting from the reduction by the
symmetry group $O(K)$ of some symplectic dynamics on cotangent bundle $T^{*}O(K)$.
\end{remark}
\begin{remark}\label{pr.2.5}
The dynamics (\ref{eq.2.1}) admits a family of Hamilton-Poisson realizations parametrized by the Lie group $SL(2;\R)$.
More exactly, $(\R^3,\{\cdot,\cdot\}_{\alpha,\beta},H^{\gamma,\delta})$ is a Hamilton-Poisson realization of the
dynamics (\ref{eq.2.1}), where $\left[ {\begin{array}{*{20}c}
   \alpha & \beta  \\
   \gamma & \delta  \\
\end{array}} \right]\in SL(2,\R)$, the bracket $\{\cdot,\cdot\}_{\alpha,\beta}$ is defined by
$$\{f,g\}_{\alpha,\beta}:=-\nabla C^{\alpha,\beta}\cdot(\nabla f\times\nabla g),$$ for any $f,g\in C^\infty(\R^3,\R)$,
and the functions $C^{\alpha,\beta},H^{\gamma,\delta}\in C^\infty(\R^3,\R)$ are given by:
$$C^{\alpha,\beta}(u):=\dfrac{1}{2}[u^{T}(\alpha K+\beta A)u]+\beta u^T \bold{a},$$
$$H^{\gamma,\delta}(u):=\dfrac{1}{2}[u^{T}(\gamma K+\delta A)u]+\delta u^T \bold{a},$$
for any $u\in\R^3$.
\end{remark}

\section{Normal forms of quadratic and homogeneous Hamiltonian systems on $(\mathfrak{o}(K))^{*}$}

In this section we compute explicitly the normal form of a general quadratic and homogeneous Hamiltonian system on $(\mathfrak{o}(K))^{*}$. More exactly, we show that if there exists $\alpha,\beta\in\mathbb{R}$, $\beta\neq 0$, such that $K^{\alpha,\beta}:=\alpha A+\beta K$ is positive definite, then the quadratic and homogeneous Hamiltonian system $(\R^3,\{\cdot,\cdot\}_{K},H_{A})$, is linearly equivalent to the relaxed free rigid body dynamics.

Let us start with the case when the symmetric matrix $K$ is positive definite.
\begin{proposition}\label{pr.5.1}
If the matrix $K$ is positive definite, then the system \eqref{eq.2.2} is linearly equivalent to
the dynamical system:
\begin{equation}\label{hat}
\dot w=w\times(\hat{A}w),\ w\in\R^3,
\end{equation}
where the symmetric matrix $\hat{A}$ is given by $\hat{A}=:L^{-1}A(L^{-1})^{T}$, and $L\in GL(3,\mathbb{R})$ such that $K=LL^{T}$. 
\end{proposition}
\begin{proof}
Let us first make the notation $L^{-T}:=(L^{-1})^{T}$. Next, we show that:
\begin{equation*}
u(t)=\operatorname{det} (L^{-1})L^{-T}w(t),
\end{equation*}
where $t\mapsto u(t)$ is a solution of the dynamical system
\eqref{eq.2.2}:
$$\dot u=(Ku)\times(Au),$$
and respectively $t\mapsto w(t)$, is a solution of the dynamical
system
$$\dot w=w\times(\hat{A}w).$$
To prove this assertion, note first that:
\begin{align*}
(Ku)\times (Au)&=(LL^{T}u)\times (Au)\\
&=[LL^{T}\operatorname{det}(L^{-1})L^{-T}w]\times [A\operatorname{det}(L^{-1})L^{-T}w]\\
&=(\operatorname{det}(L^{-1}))^2[(Lw)\times (AL^{-T}w)]\\
&=(\operatorname{det}(L^{-1}))^2[(Lw)\times (L\hat{A}w)]\\
&=(\operatorname{det}(L^{-1}))^{2} \operatorname{det}(L)L^{-T}[w\times (\hat{A}w)]\\
&=\operatorname{det}(L^{-1})L^{-T}[w\times (\hat{A}w)].
\end{align*}
Hence,
\begin{align*}
\dot u=(Ku)\times (Au)&\Leftrightarrow \operatorname{det} (L^{-1})L^{-T}\dot w=\operatorname{det}(L^{-1})L^{-T}[w\times (\hat{A}w)]\\
&\Leftrightarrow\dot w=w\times(\hat{A}w).
\end{align*}
\end{proof}
\begin{remark}
Note that since $K$ is a real, symmetric and positive definite matrix, by Cholesky decomposition we always get a lower triangular matrix $L\in GL(3,\mathbb{R})$ such that $K=LL^{T}$.
\end{remark}

Next proposition study the case when the matrix $K$ is not positively definite, but there exists $\alpha,\beta\in\mathbb{R}$ such that $K^{\alpha,\beta}:=\alpha A+\beta K$ is positive definite. One shows that this case can be reduced to the above studied case.

\begin{proposition}\label{pr.5.2}
If there exists $\alpha,\beta\in\mathbb{R}$, $\beta\neq 0$ such that the symmetric matrix $K^{\alpha,\beta}$ is positive definite, then the system \eqref{eq.2.2} is homothetically equivalent to the dynamical system:
$$\dot p=(K^{\alpha,\beta}p)\times(Ap),\ p\in\R^3.$$
\end{proposition}
\begin{proof}
We show that:
\begin{equation*}
u(t)=\beta p(t),
\end{equation*}
where $t\mapsto u(t)$ is a solution of the dynamical system
\eqref{eq.2.2}:
$$\dot u=(Ku)\times(Au),$$
and respectively $t\mapsto p(t)$, is a solution of the dynamical
system
$$\dot p=(K^{\alpha,\beta}p)\times(Ap).$$
To prove this assertion, note first that:
\begin{align*}
(Ku)\times (Au)&=(K\beta p)\times (A\beta p)\\
&=\beta[(\beta K p)\times (A p)]\\
&=\beta[((\alpha A+\beta K) p)\times (A p)]\\
&=\beta[(K^{\alpha,\beta} p)\times (A p)].
\end{align*}
Hence,
\begin{align*}
\dot u=(Ku)\times (Au)&\Leftrightarrow \beta \dot p=\beta[(K^{\alpha,\beta} p)\times (A p)]\\
&\Leftrightarrow\dot p=(K^{\alpha,\beta}p)\times(Ap).
\end{align*}
\end{proof}

\begin{proposition}\label{pr.5.3}
The system \eqref{hat} is orthogonally equivalent to
the dynamical system:
\begin{equation}\label{qhdiag}
\dot v=v\times(\hat{D}v),\ v\in\R^3,
\end{equation}
where $\hat{D}:=\operatorname{diag}(\lambda_1,\lambda_2,\lambda_3)\in\mathfrak{gl}(3,\mathbb{R})$, and $R\in O(\operatorname{Id})$ is an
orthogonal matrix such that $R^T\hat{A}R=\hat{D}$.
\end{proposition}
\begin{proof}
Note first that the equation \eqref{hat} is of the type $(\R^3,\{\cdot,\cdot\}_{K},H_{A})$, for $K=\operatorname{Id}$ and $A=\hat{A}$. The proof follows by Proposition \eqref{pr.5.1}, for $L=R\in O(\operatorname{Id})$ such that $R^T\hat{A}R=\hat{D}$, where $\hat{D}:=\operatorname{diag}(\lambda_1,\lambda_2,\lambda_3)\in\mathfrak{gl}(3,\mathbb{R})$.
\end{proof}

\begin{remark}\label{rem.5.2}
If there exists $\alpha,\beta\in\mathbb{R}$, $\beta\neq 0$ such that the symmetric matrix $K^{\alpha,\beta}$ is positive definite, then the dynamical system \eqref{eq.2.2} is
linearly equivalent to the system \eqref{qhdiag}, where $\hat{D}:=\operatorname{diag}(\lambda_1,\lambda_2,\lambda_3)$ is the diagonal form
of the matrix $\hat{A}\in Sym(3)$ which generates the system \eqref{hat}.
Note that, using coordinates, the system \eqref{qhdiag} becomes:
$$\left\{\begin{array}{l}
\dot x_1=(\lambda_3-\lambda_2)x_2 x_3,\\
\dot x_2=(\lambda_1-\lambda_3)x_1 x_3,\\
\dot x_3=(\lambda_2-\lambda_1)x_1 x_2.\\
\end{array}\right.
$$
\end{remark}
\bigskip
Hence, we proved that if there exists $\alpha,\beta\in\mathbb{R}$, $\beta\neq 0$ such that the symmetric matrix $K^{\alpha,\beta}:=\alpha A+\beta K$ is positive definite, then the quadratic and homogeneous Hamiltonian system $(\R^3,\{\cdot,\cdot\}_{K},H_{A})$ describes actually the relaxed free rigid body dynamics.


\section{Normal forms of quadratic Hamiltonian systems on $(\mathfrak{o}(K))^{*}$}
In this section we compute explicitly the normal forms of a general quadratic Hamiltonian system on $(\mathfrak{o}(K))^{*}$. We construct explicitly an affinely equivalent system to the dynamical system \eqref{eq.2.1}, and show that this system describes the relaxed free rigid body dynamics with three linear controls.

Let us start with the case when the symmetric matrix $K$ is positive definite. 
\begin{proposition}\label{pr.6.1}
If the matrix $K$ is positive definite, then the system \eqref{eq.2.1} is linearly equivalent to
the dynamical system:
\begin{equation}\label{hata}
\dot w=w\times(\hat{A}w+\hat{\bold{a}}),\ w\in\R^3,
\end{equation}
where the symmetric matrix $\hat{A}$ is given by $\hat{A}=:L^{-1}A(L^{-1})^{T}$, $L\in GL(3,\mathbb{R})$ such that $K=LL^{T}$, and $\hat{\bold{a}}:=\operatorname{det}(L)L^{-1}\bold{a}$.
\end{proposition}
\begin{proof}
Recall that $L^{-T}:=(L^{-1})^{T}$. We will show that:
\begin{equation*}
u(t)=\operatorname{det} (L^{-1})L^{-T}w(t),
\end{equation*}
where $t\mapsto u(t)$ is a solution of the dynamical system
\eqref{eq.2.1}:
$$\dot u=(Ku)\times(Au+\bold{a}),$$
and respectively $t\mapsto w(t)$, is a solution of the dynamical
system
$$\dot w=w\times(\hat{A}w+\hat{\bold{a}}).$$
To prove this assertion, note first that:
\begin{align*}
(Ku)\times (Au+\bold{a})&=(LL^{T}u)\times (Au+\bold{a})\\
&=[LL^{T}\operatorname{det}(L^{-1})L^{-T}w]\times [A\operatorname{det}(L^{-1})L^{-T}w+\bold{a}]\\
&=(\operatorname{det}(L^{-1}))^2[(Lw)\times (AL^{-T}w+\operatorname{det}(L)\bold{a})]\\
&=(\operatorname{det}(L^{-1}))^2[(Lw)\times (L\hat{A}w+\operatorname{det}(L)LL^{-1}\bold{a}))]\\
&=(\operatorname{det}(L^{-1}))^{2} \operatorname{det}(L)L^{-T}[w\times (\hat{A}w+\operatorname{det}(L)L^{-1}\bold{a})]\\
&=\operatorname{det}(L^{-1})L^{-T}[w\times (\hat{A}w+\hat{\bold{a}})].
\end{align*}
Hence,
\begin{align*}
\dot u=(Ku)\times (Au)&\Leftrightarrow \operatorname{det} (L^{-1})L^{-T}\dot w=\operatorname{det}(L^{-1})L^{-T}[w\times (\hat{A}w+\hat{\bold{a}})]\\
&\Leftrightarrow\dot w=w\times(\hat{A}w+\hat{\bold{a}}).
\end{align*}
\end{proof}
In the case when the matrix $K$ is not positive definite, but there exists $\alpha,\beta\in\mathbb{R}$ such that $\alpha K+\beta A$ is positive definite, we use the fact that the system \eqref{eq.2.1} is of the general type \eqref{general}, which by Proposition \eqref{cdd} is affinely equivalent to a dynamical system of the type \eqref{hata}. 

\begin{proposition}\label{pr.6.3}
The system \eqref{hata} is orthogonally equivalent to
the dynamical system:
\begin{equation}\label{hatadiag}
\dot v=v\times(\hat{D}v+\hat{\bold{d}}),\ v\in\R^3,
\end{equation}
where $\hat{D}:=\operatorname{diag}(\lambda_1,\lambda_2,\lambda_3)\in\mathfrak{gl}(3,\mathbb{R})$, $R\in O(\operatorname{Id})$ is an
orthogonal matrix such that $R^T\hat{A}R=\hat{D}$, and $\hat{\bold{d}}:=\operatorname{det}(R)R^{T}\hat{\bold{a}}$.
\end{proposition}
\begin{proof}
Note first that the equation \eqref{hata} is of the type $(\R^3,\{\cdot,\cdot\}_{K},H_{(A,\bold{a})})$, where $K=\operatorname{Id}$, $A=\hat{A}$ and $\bold{a}=\hat{\bold{a}}$. The proof follows by Proposition \eqref{pr.6.1}, for $L=R\in O(\operatorname{Id})$ such that $R^T\hat{A}R=\hat{D}$, where $\hat{D}:=\operatorname{diag}(\lambda_1,\lambda_2,\lambda_3)\in\mathfrak{gl}(3,\mathbb{R})$, and $\hat{\bold{d}}:=\operatorname{det}(R)R^{T}\hat{\bold{a}}$.
\end{proof}
\begin{remark}
If there exists $\alpha,\beta\in\mathbb{R}$ such that $\alpha K+\beta A$ is positive definite, then the dynamical system \eqref{eq.2.2} is affinely equivalent to the system \eqref{hatadiag}, where $\hat{D}:=\operatorname{diag}(\lambda_1,\lambda_2,\lambda_3)$ is the diagonal form of the matrix $\hat{A}\in Sym(3)$ which generates the corresponding quadratic Hamiltonian system of the type \eqref{hata}.
Using coordinates, the system \eqref{hatadiag} becomes:
$$\left\{\begin{array}{l}
\dot x_{1}=(\lambda_3-\lambda_2)x_{2}x_{3}+d_{3}x_{2}-d_{2}x_{3},\\
\dot x_{2}=(\lambda_1-\lambda_3)x_{1}x_{3}-d_{3}x_{1}+d_{1}x_{3},\\
\dot x_{3}=(\lambda_2-\lambda_1)x_{1}x_{2}+d_{2}x_{1}-d_{1}x_{2},\\
\end{array}\right.
$$
where $(d_1,d_2,d_3)$ are the coordinates of $\hat{\bold{d}}$.
\end{remark}
\bigskip
Hence, we proved that if there exists $\alpha,\beta\in\mathbb{R}$ such that $\alpha K+\beta A$ is positive definite, then the quadratic Hamiltonian system $(\R^3,\{\cdot,\cdot\}_{K},H_{(A,\bold{a})})$ describes actually the relaxed free rigid body dynamics with three linear controls.

\section{A large class of three dimensional quadratic Hamiltonian systems}

In this section we give a unified and also a generalized formulation of the results obtained in the previous sections, by considering a general quadratic Hamiltonian system on a natural extension of $(\mathfrak{o}(K))^*$, where the pair $(K,\bold{k})\in Sym(3)\times \mathbb{R}^3$ is supposed to be fixed.

Let us now start by introducing the Poisson manifold $(\mathbb{R}^{3},\{\cdot,\cdot\}_{(K,\bold{k})})$, where the Poisson bracket $\{\cdot,\cdot\}_{(K,\bold{k})}$ is defined by:
$$\{f,g\}_{(K,\bold{k})}:=-\nabla C_{(K,\bold{k})}\cdot(\nabla f\times\nabla g),$$ 
for any $f,g\in C^\infty(\R^3,\R)$, and the smooth function $C_{(K,\bold{k})}\in C^\infty(\R^3,\R)$ is given by
$$C_{(K,\bold{k})}(u):=\dfrac{1}{2}u^{T}Ku+u^{T}\bold{k}.$$
Note that by the definition of the Poisson bracket, it follows that the smooth function $C_{(K,\bold{k})}$ is a Casimir invariant.

Consequently, a quadratic Hamiltonian system on $(\mathbb{R}^{3},\{\cdot,\cdot\}_{(K,\bold{k})})$, is generated by a smooth function $H_{(A,\bold{a})}\in C^\infty(\R^3,\R)$, given by $$H_{(A,\bold{a})}(u):=\dfrac{1}{2}u^{T}Au+u^{T}\bold{a},$$
where $A\in Sym(3)$ is an arbitrary real symmetric matrix, and $\bold{a}\in\mathbb{R}^3$. 

Hence, the associated Hamiltonian system is given by
\begin{equation}\label{general}
\dot u=(Ku+\bold{k})\times(Au+\bold{a}),\ u\in\R^3.
\end{equation}

\begin{remark}\label{rok}
The dynamics (\ref{general}) admits a family of Hamilton-Poisson realizations parametrized by the Lie group $SL(2;\R)$.
More exactly, $(\R^3,\{\cdot,\cdot\}_{\alpha,\beta},H^{\gamma,\delta})$ is a Hamilton-Poisson realization of the
dynamics (\ref{general}), where $\left[ {\begin{array}{*{20}c}
   \alpha & \beta  \\
   \gamma & \delta  \\
\end{array}} \right]\in SL(2,\R)$, the bracket $\{\cdot,\cdot\}_{\alpha,\beta}$ is defined by
$$\{f,g\}_{\alpha,\beta}:=-\nabla C^{\alpha,\beta}\cdot(\nabla f\times\nabla g),$$ for any $f,g\in C^\infty(\R^3,\R)$,
and the functions $C^{\alpha,\beta},H^{\gamma,\delta}\in C^\infty(\R^3,\R)$ are given by:
$$C^{\alpha,\beta}(u):=\dfrac{1}{2}[u^{T}(\alpha K+\beta A)u]+u^T(\alpha \bold{k}+ \beta \bold{a}),$$
$$H^{\gamma,\delta}(u):=\dfrac{1}{2}[u^{T}(\gamma K+\delta A)u]+ u^T(\gamma \bold{k}+ \delta\bold{a}),$$
for any $u\in\R^3$.
\end{remark}

Next theorem shows that for $K$ positive definite, a quadratic Hamiltonian system on $(\mathbb{R}^{3},\{\cdot,\cdot\}_{(K,\bold{k})})$ is affinely equivalent to a quadratic Hamiltonian system on $(\mathfrak{so}(3))^{*}$, and consequently equivalent to the relaxed free rigid body dynamics with three linear controls, as already proved in the above section.

\begin{theorem}\label{th.6.1}
If the matrix $K$ is positive definite, then the system \eqref{general} is affinely equivalent to
the dynamical system:
\begin{equation*}
\dot w=w\times(\hat{A}w+\hat{\bold{a}}),\ w\in\R^3,
\end{equation*}
where the symmetric matrix $\hat{A}$ is given by $\hat{A}=:L^{-1}A(L^{-1})^{T}$, $L\in GL(3,\mathbb{R})$ such that $K=LL^{T}$, and $\hat{\bold{a}}:=\operatorname{det}(L)L^{-1}(\bold{a}-AK^{-1}\bold{k})$.
\end{theorem}
\begin{proof}
Recall that $L^{-T}:=(L^{-1})^{T}$. We will show that:
\begin{equation*}
u(t)=\operatorname{det} (L^{-1})L^{-T}w(t)-K^{-1}\bold{k},
\end{equation*}
where $t\mapsto u(t)$ is a solution of the dynamical system
\eqref{general}:
$$\dot u=(Ku+\bold{k})\times(Au+\bold{a}),$$
and respectively $t\mapsto w(t)$, is a solution of the dynamical
system
$$\dot w=w\times(\hat{A}w+\hat{\bold{a}}).$$
To prove this assertion, note first that:
\begin{align*}
(Ku+\bold{k})&\times (Au+\bold{a})=[K(\operatorname{det} (L^{-1})L^{-T}w-K^{-1}\bold{k})+\bold{k}]\times \\
&\times [A(\operatorname{det} (L^{-1})L^{-T}w-K^{-1}\bold{k})+\bold{a}]\\
&=(LL^{T}\operatorname{det}(L^{-1})L^{-T}w)\times [A\operatorname{det}(L^{-1})L^{-T}w+(\bold{a}-AK^{-1}\bold{k})]\\
&=(\operatorname{det}(L^{-1})Lw)\times [A\operatorname{det}(L^{-1})L^{-T}w+(\bold{a}-AK^{-1}\bold{k})]\\
&=(\operatorname{det}(L^{-1})Lw)\times [\operatorname{det}(L^{-1})AL^{-T}w+\operatorname{det}(L^{-1})L\hat{\bold{a}}]\\
&=(\operatorname{det}(L^{-1}))^{2}[(Lw)\times (AL^{-T}w+L\hat{\bold{a}})]\\
&=(\operatorname{det}(L^{-1}))^{2}[(Lw)\times (L\hat{A}w+L\hat{\bold{a}})]\\
&=(\operatorname{det}(L^{-1}))^{2}\operatorname{det}(L)L^{-T}[w\times (\hat{A}w+\hat{\bold{a}})]\\
&=\operatorname{det}(L^{-1})L^{-T}[w\times (\hat{A}w+\hat{\bold{a}})].
\end{align*}
Hence,
\begin{align*}
\dot u=(Ku+\bold{k})\times (Au+\bold{a})&\Leftrightarrow (\operatorname{det} (L^{-1}))L^{-T}\dot w=(\operatorname{det}(L^{-1}))L^{-T}[w\times (\hat{A}w+\hat{\bold{a}})]\\
&\Leftrightarrow\dot w=w\times(\hat{A}w+\hat{\bold{a}}).
\end{align*}
\end{proof}

Using the Remark \eqref{rok}, note that if there exists $\alpha,\beta\in\R$ such that $\alpha K+\beta A$ is positive definite, then for any $\gamma,\delta\in\R$ such that $\left[ {\begin{array}{*{20}c}
   \alpha & \beta  \\
   \gamma & \delta  \\
\end{array}} \right]\in SL(2,\R)$, the system \eqref{general} can be written in the equivalent form
\begin{equation}\label{egs}
\dot u=(P^{\alpha,\beta}u+\bold{p^{\alpha,\beta}})\times(A^{\gamma,\delta}u+\bold{a^{\gamma,\delta}}),
\end{equation}
where $P^{\alpha,\beta}=\alpha K+\beta A$, $\bold{p^{\alpha,\beta}}=\alpha \bold{k}+\beta\bold{a}$, $A^{\gamma,\delta}=\gamma K+\delta A$, and $\bold{a^{\gamma,\delta}}=\gamma \bold{k}+\delta \bold{a}$. Consequently, using this remark and the Theorem \eqref{th.6.1} we obtain the following result.

\begin{proposition}\label{cdd}
If there exists $\alpha,\beta\in\R$ such that $P^{\alpha,\beta}=\alpha K+\beta A$ is positive definite, then the system \eqref{general} is affinely equivalent to the dynamical system
\begin{equation*}
\dot w=w\times(\hat{A}w+\hat{\bold{a}}),\ w\in\R^3,
\end{equation*}
where the symmetric matrix $\hat{A}$ is given by $\hat{A}=:L^{-1}A^{\gamma,\delta}(L^{-1})^{T}$, $L\in GL(3,\mathbb{R})$ such that $P^{\alpha,\beta}=LL^{T}$,  $\hat{\bold{a}}:=\operatorname{det}(L)L^{-1}[\bold{a^{\gamma,\delta}}-A^{\gamma,\delta}(P^{\alpha,\beta})^{-1}\bold{p^{\alpha,\beta}}]$, and respectively $\gamma,\delta\in\R$ such that $\left[ {\begin{array}{*{20}c}
   \alpha & \beta  \\
   \gamma & \delta  \\
\end{array}} \right]\in SL(2,\R)$.
\end{proposition}

Hence, by combining this result with Proposition \eqref{pr.6.3}, we proved that if there exists $\alpha,\beta\in\mathbb{R}$ such that the symmetric matrix $P^{\alpha,\beta}:=\alpha K+\beta A$ is positive definite, then the quadratic Hamiltonian system \eqref{general},   $(\R^3,\{\cdot,\cdot\}_{K,\bold{k}},H_{A,\bold{a}})$, describes actually the relaxed free rigid body dynamics with three linear controls. 

Since any first-order autonomous three-dimensional differential equation possessing two independent quadratic constants of motion can be written as a Hamilton-Poisson dynamical system of type \eqref{general}(eventually after a time re-parameterization; see for details \cite{tudoran}), we obtain the following consequence.

\begin{remark}
Any first-order autonomous three-dimensional differential equation possessing two independent quadratic constants of motion which admits a positive/negative definite linear combination, is affinely equivalent to the classical "relaxed" free rigid body dynamics with linear linear controls.
\end{remark}

Note that if there exists $\alpha,\beta,\gamma,\delta\in\R$ as in Proposition \eqref{cdd} such that the following relation holds true $$\bold{a^{\gamma,\delta}}-A^{\gamma,\delta}(P^{\alpha,\beta})^{-1}\bold{p^{\alpha,\beta}}=0,$$ then by Proposition \eqref{pr.5.3}, the system \eqref{general} is affinely equivalent to the classical relaxed free rigid body dynamics. 

Next result gives another sufficient condition for which the general system \eqref{general} is affinely equivalent to the classical relaxed free rigid body dynamics, via the Propositions \eqref{pr.5.1}, \eqref{pr.5.3}. 
\begin{proposition}
If there exists $\alpha,\beta\in\mathbb{R}$, $\beta\neq 0$ and $\gamma \in\R^3$ such that the symmetric matrix $K^{\alpha,\beta}=\alpha A+\beta K$ is positive definite, and $A\gamma +\bold{a}=0$, $K\gamma +\bold{k}=0$, then the system \eqref{general} is homothetically equivalent to the dynamical system:
$$\dot p=(K^{\alpha,\beta}p)\times(Ap),\ p\in\R^3.$$
\end{proposition}
\begin{proof}
We show that:
\begin{equation*}
u(t)=\beta p(t)+\gamma,
\end{equation*}
where $t\mapsto u(t)$ is a solution of the dynamical system
\eqref{general}:
$$\dot u=(Ku+\bold{k})\times(Au+\bold{a}),$$
and respectively $t\mapsto p(t)$, is a solution of the dynamical
system
$$\dot p=(K^{\alpha,\beta}p)\times(Ap).$$
To prove this assertion, note first that:
\begin{align*}
(Ku+\bold{k})\times (Au+\bold{a})&=(K\beta p+K\gamma +\bold{k})\times (A\beta p+A\gamma +\bold{a})\\
&=\beta[(\beta K p)\times (A p)]\\
&=\beta[((\alpha A+\beta K) p)\times (A p)]\\
&=\beta[(K^{\alpha,\beta} p)\times (A p)].
\end{align*}
Hence,
\begin{align*}
\dot u=(Ku+\bold{k})\times (Au+\bold{a})&\Leftrightarrow \beta \dot p=\beta[(K^{\alpha,\beta} p)\times (A p)]\\
&\Leftrightarrow\dot p=(K^{\alpha,\beta}p)\times(Ap).
\end{align*}
\end{proof}

\subsection*{Acknowledgment}
This work was supported by a grant of the Romanian National Authority for Scientific Research, CNCS - UEFISCDI, project number PN-II-RU-TE-2011-3-0103.

\end{document}